\documentclass[%
 reprint,prl,
 amsmath,amssymb,
 aps,
]{revtex4-2}

\usepackage{graphicx}
\usepackage{dcolumn}
\usepackage{bm}
\usepackage{amsthm}

\newtheorem{theorem}{Theorem}
\newtheorem{corollary}{Corollary}

\begin{document}


\title{All Or Nothing: No-Downfolding Theorems For Quantum Simulation}

\author{Troy Van Voorhis}
 \affiliation{Department of Chemistry, Massachusetts Institute of Technology, 77 Massachusetts Ave. Cambridge, MA 02139.}
 \email{tvan@mit.edu}

\date{\today}

\begin{abstract}
The physics of a quantum system with many degrees of freedom is often approximated by downfolding: most of the degrees of freedom are "folded into" a much smaller number of degrees of freedom, resulting in an effective Hamiltonian that still captures the essential physics. Approaches of this sort are particularly relevant for quantum information, where exact downfolding would allow eigenstates in a large Hilbert space to be simulated with fewer qubits. Very little work has been done to prove the existence of such an exact downfolding for general systems or even particular cases. In this letter we prove that exact quantum downfolding is impossible for what is perhaps the most commonly-used formulation of the task. Specifically, for a non-trivial Hamiltonian $\mathbf{H}(x)\equiv\mathbf{A}+x\mathbf{B}$ that depends on some parameter $x$ (e.g. an electric field, bond length or interaction strength) it is not possible to construct a lower-dimension effective Hamiltonian $\mathbf{h}(x)\equiv\mathbf{a}+x\mathbf{b}$ that exactly recovers \emph{any} of the 
eigenvalue functions $E_i(x)$ of $\mathbf{H}(x)$. We discuss several generalizations of this result and the impacts of these findings on future directions in quantum information.

\end{abstract}

\maketitle


The most recent resource estimates suggest that the number of logical qubits required for quantum simulation of numerous quantum field theories, strongly correlated molecules and solids and many body problems exceeds the capabilities of existing quantum hardware \cite{GonthierRomero22,DelgadoMiguel22,ClintonSheridan24,BluvsteinLukin24}. 
This observation has led to significant innovations in quantum downfolding \cite{BauerTroyer16,Kowalski21,HuangEvangelista23,GuntherChristandl24,AlvertisTubman25}, in which a large number of degrees of freedom are folded into a much smaller space - potentially allowing larger systems to be treated with near-term hardware\cite{PengWu20}.
Techniques of this type have a long history in physics. Effective Hamiltonians \cite{Transcorrelated,ZhangRice88,JamesJerke07,LiuZhang10,LiuYao11,BravyiLoss11,GoldmanDalibard14,ChangWagner24}, multireference models \cite{CASPT,MRMP,NEVPT,RoosMerchan96,MRCI}, coarse-graining \cite{GellmannHartle93,HolzheyWilczek94} and the decimation process of renormalization group theory \cite{Wilson75,GlazekWilson94,DMRG,Vidal07} all distill the essential physics of large, complex systems to a smaller number of degrees of freedom. 
In all cases, the assumption is that if one chooses the correct degrees of freedom, the downfolding will be qualitatively correct and well-behaved.

In this letter we consider the question of whether there is an exact quantum downfolding. That is, we consider the question of whether a qualitatively correct downfolding calculation can be considered an approximation to some (possibly unknown) reduced dimension model that agrees exactly with the full system. This question is of obvious importance to computing. It will not be possible to construct a set of controlled approximations that converge to the exact answer if there is no exact answer toward which one can converge.


In one context, exact quantum downfolding is obviously possible. Given a Hamiltonian, $\mathbf{H}$, in an $N$ dimensional Hilbert space, one can trivially construct an $M<N$ dimensional $\mathbf{h}$ that matches $M$ eigenvalues of $\mathbf{H}$: find the eigenvalues of $\mathbf{H}$ that are of interest ($E_i$) and then choose $\mathbf{h}$ to be a diagonal matrix with those eigenvalues. 

However, in many applications, one is not only interested in energies of a single Hamiltonian but rather the variation of those energies as a function of some parameter - magnetic field,  the distance between two planes, an interaction strength, etc. To capture this, consider the case where the dependence on the parameter, $x$, is linear, so that:
\begin{equation}
\mathbf{H}(x)\equiv\mathbf{A}+x\mathbf{B}    
\label{eq:Hx}
\end{equation}
for Hermitian matrices $\mathbf{A},\mathbf{B}$. Downfolding then tries to recover some of the eigenvalues of $\mathbf{H}(x)$ from a smaller matrix $\mathbf{h}(x)$.  As we will see, even this simple objective exhibits the limitations of quantum downfolding.

The eigenvalues, $E_i(x)$ of $\mathbf{H}(x)$ are the roots of the characteristic polynomial:
\begin{equation}
p_{\mathbf H} (E,x)\equiv |\mathbf{H}(x)-E \mathbf{1}|=0   
\label{eq:characteristic_polynomial}
\end{equation}
so that the eigenvalues are, by construction, algebraic functions of $x$.
One of the peculiar features of algebraic functions is that they can be multiple-valued. In practice, that means that what we think of as two distinct eigenvalue functions, $E_i(x)$ and $E_j(x)$, may in fact be two different branches of a single function.
For example, in the simple $2\times 2$ case:
\begin{equation}\begin{split}
\mathbf{H}(x)\equiv
   \left( {\begin{array}{cc}
   \epsilon & V \\
   V & -\epsilon \\
  \end{array} } \right) + x
     \left( {\begin{array}{cc}
   -\epsilon & 0 \\
   0 & \epsilon \\
  \end{array} } \right)  \\
  ~~~~~\rightarrow~~~~~
  E_{\pm}(x)=\pm \sqrt{4\epsilon^2 (1-x)^2 +V^2}
  \end{split}
  \label{eqn:Epm}
\end{equation}
it is clear that the two eigenvalues are two different sign conventions for defining the same square root function.
It is possible to demonstrate that these are a single function by analytically continuing $\mathbf{H}(x)$ to complex values of $x$. These Hamiltonians are not physically realizable (they are not, for example, Hermitian  or $\mathcal{PT}$-symmetric and have complex eigenvalues \cite{Bender07,Moiseyev11}). However, because $E_i(x)$ is an algebraic function the analytic continuation is locally well defined away from any degenerate points ($E_i(x)=E_j(x)$). In the present $2\times 2$ case, Figure~\ref{fig:Analytic_Continuation} shows $E_+(x(t))$ when $x$ moves along the path $x(t)=\sqrt{2}i(1-e^{it})$.
$x(t)$ forms a loop in the complex plane that begins and ends at the base $x=0$. The different eigenvalues at the base are termed the fiber of the multi-valued function. When traced continuously along this loop, $E_{+}(x(t=0))\rightarrow E_{+}(x=0)$ but $E_{+}(x(t=2\pi))\rightarrow E_{-}(x=0)$. If we follow $E_-(x(t))$ along the same path, we find $E_{-}(x(t=2\pi))=E_{+}(x=0)$. That is, the two eigenvalues are interchanged by moving along this path, substantiating our claim that these are, in fact, two branches of the same function.

\begin{figure}
\includegraphics[width=90mm]{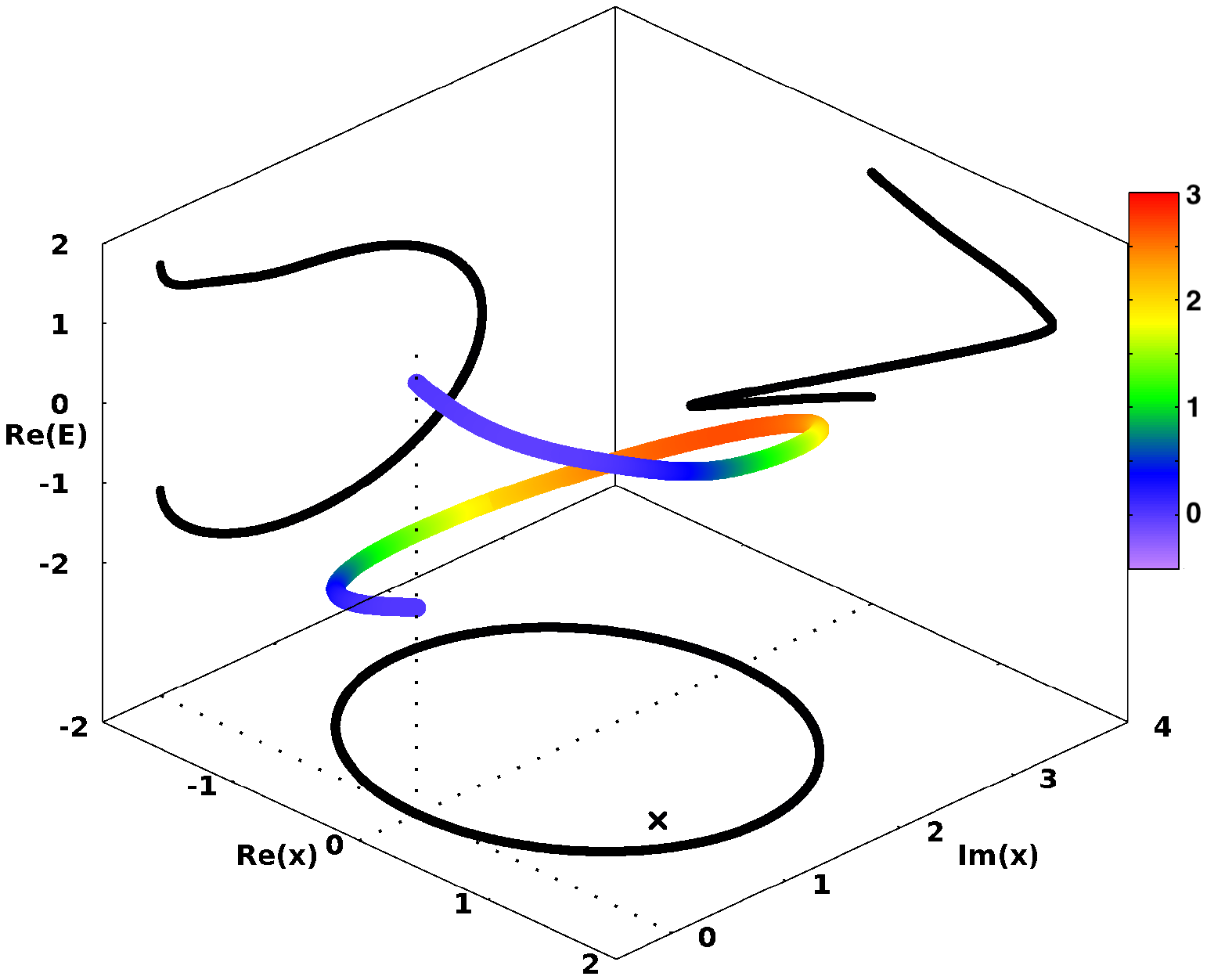}
\caption{\label{fig:Analytic_Continuation} Monodromy of eigenvalue $E_+$ from Eq.~\ref{eqn:Epm} traced along a path $x(t)$ in the complex plane, with color depicting $Im(E)$ ($\epsilon=V=1$). As $x(t)$ winds around a branch point (marked with an "X") and returns to the starting point, the energy changes from the upper eigenvalue, $E_+$, to the lower eigenvalue, $E_-$.}
\end{figure}

The process by which roots, $E_i(x)$, of the characteristic polynomial $p_{\mathbf H}(E,x)$ are shuffled when moving through loops in the complex plane is called monodromy. If $p_{\mathbf H}$ does not factorize into a product of lower order polynomials (i.e. $p_{\mathbf H}(E,x)\ne q(E,x) r(E,x)$) we say it is irreducible and we have the following important theorem:

\begin{theorem}
If the polynomial $p_{\mathbf H}(E,x)$ is irreducible then every root $E_i(x)$ of $p_{\mathbf H}$ can be connected to every other root $E_j(x)$ via monodromy.
\label{thm:Monodromy}
\end{theorem}

\begin{proof}
This theorem is proven in various textbooks \cite{Artin,Miranda, Shafarevich}. The basic structure of the proof is as follows: If the polynomial is irreducible, the associated Riemann surface of the root function $E(x)$ is a connected cover of finite degree away from isolated points of degeneracy. 
Under these conditions, the group of monodromy transformations at any base point $x$ is transitive, which means there is at least one fiber $E_i(x)$  that maps to every other $E_j(x)$. Because monodromy is reversible, every root can thus be connected to every other root by monodromy. 
\end{proof}

If all of the eigenvalues $E_i(x)$ can be connected to one another,
then each eigenvalue function in some sense encodes all eigenvalues.
$p_{\mathbf H}(E,x)$ being irreducible would thus have significant consequences for downfolding techniques aimed at identifying only some of the eigenvalues.
Now, it is clear that not every
$\mathbf{H}(x)$ has an irreducible  $p_{\mathbf H}(E,x)$. For example,
if $\mathbf{H}(x)$ is block diagonal:
\begin{eqnarray}
  \mathbf{H}(x) &\equiv&
  \left( {\begin{array}{cc}
            \mathbf{F}(x) & \mathbf{0} \\
            \mathbf{0} & \mathbf{G}(x) \\
          \end{array} } \right)
=  \left( {\begin{array}{cc}
            \mathbf{F}_0 & \mathbf{0} \\
            \mathbf{0} & \mathbf{G}_0 \\
          \end{array} } \right)
       + x
     \left( {\begin{array}{cc}
            \mathbf{F}_1 & \mathbf{0} \\
            \mathbf{0} & \mathbf{G}_1 \\
          \end{array} } \right)  \nonumber\\  
&\rightarrow&~~p_{\mathbf H}(E,x)=p_{\mathbf F}(E,x) p_{\mathbf G}(E,x)
\end{eqnarray} 
This block diagonal structure reproduces what we expect when the
system has a symmetry -
changing the parameter $x$ mixes states within
each block (states that have the same symmetry as one another) but
does not mix states in different blocks (those that have different
symmetry). One can further generalize this case to include matrices
that are similar to a block diagonal matrix 
(i.e. those that can be written as $\mathbf{H}(x)=\mathbf{U}^T(x) \mathbf{\tilde H}(x)
\mathbf{U}(x) $ where $\mathbf{\tilde H}$ is block diagonal and $\mathbf{U}$ is unitary), which
corresponds to the same blocked Hamiltonian in a basis
of states that have mixed symmetry. As it turns out, this is the only situation in which
$p_{\mathbf H}(E,x)$ is reducible:

\begin{theorem}
  If $p_{\mathbf H}(E,x)$ is reducible then $\mathbf{H}(x)$ is similar to a block diagonal matrix.
  \label{thm:irreducible}
 \end{theorem}

 \begin{proof}
   If $p_{\mathbf H}(E,x)=p_{\mathbf F}(E,x) p_{\mathbf G}(E,x)$ then
   $p_{\mathbf F}(E,x) $ is real-rooted of some total degree, $M$.
   Any polynomial of this type can
   be written as the characteristic polynomial of some $M\times M$
   symmetric matrix $\mathbf{F}(x)\equiv\mathbf{F}_0+x \mathbf{F}_1$ \cite{Hanselka17}.
   Applying the same logic to $p_{\mathbf G}(E,x)$ shows that it is
   the characteristic polynomial of some symmetric $(N-M)\times(N-M)$
   matrix $\mathbf{G}(x)\equiv\mathbf{G}_0+x \mathbf{G}_1$. Then
   \begin{equation}
     \mathbf{\tilde H}(x) \equiv
     \left( {
         \begin{array}{cc}
           \mathbf{F}(x) & \mathbf{0} \\
           \mathbf{0} & \mathbf{G}(x) \\
         \end{array} } \right)
   \end{equation}
   satisfies $p_{\mathbf{\tilde H}}(E,x)=p_{\mathbf H}(E,x)$ which
   implies the two polynomial matrices are similar.
 \end{proof}

Reducibility by similarity is relatively trivial -  the characteristic polynomial
is reducible because the entire matrix can be exactly reduced to two
smaller matrices via an innate symmetry.
Unless $\mathbf{H}(x)$ is trivial in this way, $p_{\mathbf
  {H}}(E,x)$ will be irreducible.

Theorems~\ref{thm:Monodromy} and \ref{thm:irreducible} are the central results of this letter. They
lead immediately to the following corollary for quantum downfolding:

\begin{corollary}
No eigenvalue function, $e_i(x)$ of an $M$-dimensional Hamiltonian,
$\mathbf{h}(x)$, can equal any of the eigenvalue functions, $E_j(x)$,
of an $N$-dimensional Hamiltonian $\mathbf{H}(x)$  ($N>M$) except for
the trivial case in which
$\mathbf{H}(x)$ is similar to a block diagonal matrix with one block
of dimension $M$ or smaller.  
\label{cor:NoDownfolding}
\end{corollary}

\begin{proof}
If $\mathbf{H}(x)$ is not similar to a block matrix of the stated form
then every irreducible factor of $p_{\mathbf{H}}$ has degree at least 
$M+1$. That means that every eigenvalue function $E_j(x)$ can be
connected to at least $M+1$ eigenvalue functions by monodromy. By construction, every $e_i(x)$ can be connected to at most $M$ distinct eigenvalues by monodromy. Therefore the two functions cannot be equal.
\end{proof}

It is somewhat difficult to envision physical scenarios in which $p_{\mathbf H}$ is reducible and yet the $\mathbf{U}$ that brings it to block diagonal form  depends meaningfully on $x$. 
If $\mathbf{U}$ must depend on $x$, that would imply that while $\mathbf{H}(x)$ has an intrinsic symmetry, the symmetry character of the basis cannot be chosen to match the symmetry of $\mathbf{H}(x)$ - which would be unusual.
It is much easier to understand the situation in which $\mathbf{U}$ can be chosen constant - a situation we call \emph{strict} similarity.
Categorizing pairs of matrices $(\mathbf{A},\mathbf{B})$ up to strict
similarity is a wild problem \cite{Nazarova74,Drozd77} and thus strict similarity can be established only under certain conditions
\cite{Kippenhahn51,MotzkinTaussky53,Albert1958,Shapiro1,Shapiro2,Shapiro3,Friedland1983, LeBruyn1997,Shavarovskii04}.
In this context, it is worth noting that in the most common downfolding scenario one is interested in obtaining the
lowest few eigenvalues, in which case strict similarity is required: 

\begin{theorem}
  If $p_{\mathbf H}(E,x)= p_{\mathbf F}(E,x) p_{\mathbf G}(E,x)$ is
  reducible and all the roots of $p_{\mathbf F}$ are strictly less than the
  roots of $p_{\mathbf G}$ for some value of $x$ then $\mathbf{H}(x)$ is strictly
  similar to a block diagonal matrix.
  \label{thm:strict}
 \end{theorem}

 \begin{proof}
   We know from Theorem~\ref{thm:irreducible} that $\mathbf{H}(x)$ is
   similar to a block diagonal matrix $\mathbf{\tilde{H}}(x)$ with
   blocks $\mathbf{F}(x)$ and $\mathbf{G}(x)$ that depend linearly on $x$. Call
   the roots of $p_{\mathbf F}$ ($p_{\mathbf G}$) $E_i^F(x)$
   ($E_i^G(x)$ and suppose that at $x=x_0$ we have that
   $E_i^F(x_0)<E_j^G(x_0) ~\forall~ i,j$. Perform unitary
   transformations on  $\mathbf{H}(x)$ and $\mathbf{\tilde{H}}(x)$ so
   that both are diagonal with the eigenvalues in ascending order at
   $x_0$. Then, examine $\sum_i E_i^F(x_0+\delta
   x)$ for both $\mathbf{H}$ and $\mathbf{\tilde{H}}$. We know
   these two sums must be the same because the matrices are similar.
   A simple calculation shows that for $\mathbf{\tilde{H}}(x)$ we
   have:  $\sum_i E_i^F(x_0+\delta
   x)= Tr(\mathbf{F}_0)+\delta x
   Tr(\mathbf{F}_1)$. That is, the sum is linear in $\delta x$. For
   $\mathbf{H}(x)$ perturbation theory gives:
   \begin{eqnarray*}
     \sum_i E_i^F(x_0+\delta x)= \sum_{i\in F} E_i^F(x_0)+\delta x
     \sum_{i\in F} B_{ii} \\ 
     + \delta x^2 \sum_{i\in F,j\in G}
     \frac{|B_{ij}|^2}{E_i^F(x_0)-E_j^G(x_0)}+...
   \end{eqnarray*}
For the two sums to be equal the quadratic term must be zero. Because
$E_i^F(x_0)<E_j^G(x_0) ~\forall~ i,j$ each term in the quadratic sum is
nonpositive. To obtain a sum of zero, then, each term must be
identically zero, from which we conclude that $B_{ij}=0 ~\forall~ i\in
F,j\in G$. That is, the $FG$ block of $\mathbf{H}(x)$ is zero. We thus
conclude that diagonalizing $\mathbf{H}(x)$ at $x_0$ results in a
block diagonal $\mathbf{H}(x)$ for all $x$. That is, $\mathbf{H}(x)$
is strictly similar to a block diagonal matrix.
 \end{proof}

We thus conclude that for the most common downfolding scenario it is indeed the familiar situation of strict similarity (and not just similarity) that is required:
\begin{corollary}
The eigenvalue functions, $e_i(x)$ of an $M$-dimensional Hamiltonian,
$\mathbf{h}(x)$, cannot reproduce the $M$ lowest eigenvalue functions, $E_j(x)$,
of an $N$-dimensional Hamiltonian $\mathbf{H}(x)$  ($N>M$) except for
the trivial case in which
$\mathbf{H}(x)$ is strictly similar to a block diagonal matrix with one block
equal to $\mathbf{h}(x)$.  
\label{cor:NoRenormalization}
\end{corollary}

\begin{proof}
The proof is a straightforward application of Corollary~\ref{cor:NoDownfolding} and
Theorem~\ref{thm:strict}.
\end{proof}

Note that the above results do not hinge on monodromy being practical. As an analogy, if two functions ($e_i(x)$ and $E_j(x)$) had different formal power series expansions, we could conclude that the two functions are not equal \emph{even if the power series do not converge}. The situation with monodromy is strictly analogous. The analytic continuation implied by monodromy is likely to be unstable and computationally inefficient for many $\mathbf{H}(x)$. The above theorems still demonstrate that these unknown paths are different for $e_i(x)$ and $E_j(x)$ for non-trivial $\mathbf{H}(x)$. Thus, though they may be prohibitively difficult to locate, the existence of distinct monodromy means it is impossible for the two functions to be equal. 

The results above generalize with very little modification to the case of arbitrary polynomial dependence on $x$.
Specifically, generalize $\mathbf{H}(x)$ to a Hermitian $N\times N$ finite-degree polynomial matrix in $x$:
\begin{equation}
\mathbf{H}(x)\equiv\mathbf{A}+\sum_{k=1}^{K} x^k ~\mathbf{B}_k
\label{eq:H_poly}
\end{equation}
and consider the possibility of downfolding it to a Hermitian $M\times M$ ($M<N$) finite-degree polynomial matrix:
\begin{equation}
\mathbf{h}(x)\equiv\mathbf{a}+\sum_{l=1}^{L} x^l ~\mathbf{b}_l.
\label{eq:h_poly}
\end{equation}
Theorem~\ref{thm:Monodromy} applies to this case without modification
– as long as the characteristic polynomial $p_{\mathbf H}(E,x)$ is
irreducible, its roots can all be connected by monodromy. For the
analogous Theorem~\ref{thm:irreducible} for any polynomial matrix
$\mathbf{H}(x)$ with a reducible $p_{\mathbf{H}} (x,E)$,
$\mathbf{H}(x)$ is similar to a block-diagonal matrix where each block
is also a polynomial matrix in $x$ \cite{Hanselka17}. While in the
linear case the block polynomials can be chosen to have
the same degree as $\mathbf{H}(x)$, in the polynomial case some
block(s) could be polynomials of a different degree than the original
polynomial matrix (e.g. $K$ in Eq.~\ref{eq:H_poly} might not equal $L$
in Eq.~\ref{eq:h_poly}).
Finally, for Theorem~\ref{thm:strict}, one performs a
perturbative expansion to order $L+1$ (i.e. one higher than the order
of the polynomial matrix) to demonstrate that similarity must be
strict.
Combining these results, one recovers the analogous corollaries for any non-trivial polynomial $\mathbf{H}(x)$: none of its eigenvalue functions can equal any of the eigenvalue functions of $\mathbf{h}(x)$ if $M<N$. Thus, even for arbitrary polynomial dependence on $x$, exact global downfolding is still impossible. 

We can further generalize this by choosing some invertible function $x(s)$ (with inverse $s(x)$) and then defining:
\begin{equation}
\mathbf{\bar H}(s)\equiv\mathbf{H}(x(s))\equiv\mathbf{A}+\sum_{k=1}^{K} x(s)^k ~\mathbf{B}_k
\label{eq:H_non-poly}
\end{equation}
Choosing $x(s)\equiv e^{-\alpha s}$, for example, would give an expansion akin to a discrete Laplace transform expansion.
In this case it is also impossible for a non-trivial $\mathbf{\bar H}(s)$ to have an eigenvalue function $\bar E_i(s)$ equal to $\bar e_j(s)$ derived from a smaller $\mathbf{\bar h}(s)$. If it were, then $\bar E_i(s(x))$ (the eigenvalue of the polynomial matrix $\mathbf{H}(x)$)  would be equal to $e_j(x)=\bar e_i(s(x))$ (the eigenvalue of a polynomial matrix $\mathbf{h}(x)$) in contradiction to the above results.

The fundamental limitation of the preceding formulations is that the matrix functions involved are all single-valued, whereas the eigenvalue function, $E_k(x)$, is multivalued.
The functions, $e_j(x)$ we are using to approximate $E_k(x)$ are also multivalued, but when $M<N$ the approximate functions have the wrong number of values at the point $x$ (i.e. the fiber has the wrong dimension). A single-valued transformation like a polynomial cannot alter this situation. As result, one expects the preceding no-downfolding theorems should apply to any single-valued $\mathbf{h}(x)$, though such a result would be difficult to prove.

By contrast, if we allow $\mathbf{h}(x)$ to be a (multi-valued) \emph{algebraic} function, exact downfolding once again becomes trivial. Specifically, one could choose $\mathbf{h}(x)$ to be a diagonal matrix with the chosen eigenvalues $E_i(x)$ as the diagonal elements. This type of exact downfolding gives a formal basis to non-perturbative methods like renormalization group \cite{GlazekWilson94,Wegner94} and canonical transformation theory \cite{YanaiChan06}. However, this trivial framing of the problem is unsatisfying because it is self-referential: solving for $E_i(x)$ is itself the original goal.

One would obviously prefer an exact algebraic downfolding in which each of the algebraic functions in $\mathbf{h}(x)$ is simpler to obtain than the full $E_i(x)$. In this sense, one could hope to 'unravel' the complexity of solving the full eigenvalue problem into a series of smaller problems. There has not been much work on the question of when this is possible, and the question is hampered by the fact that almost every high-order polynomial cannot be written as a composition of lower-order polynomials \cite{DecomposablePolynomials}. However, an exact unraveling in terms of algebraic functions is not completely impossible and approaches of this type are worthy of future study.

The most widely used solution to this issue is to make the downfolded Hamiltonian energy dependent (i.e. $\mathbf{h}(x,E)$). This formulation is the basis of many downfolding approaches\cite{Lowdin63,Morris94,BravyiLoss11}, including those based on many-body Green's functions \cite{Hedin,GWReview,DMFT,scGW}. These approaches do not suffer from the same limitations as the energy-independent frameworks described above and downfolding is quite clearly possible with an energy dependent effective Hamiltonian. However, this comes at a cost, in that the multi-valued nature of $E_i(x)$ has been "rolled up" into the energy, $E$. As a result, one is typically not describing a range of states with a single $\mathbf{h}(x,E)$ but rather assigning each state its own specific $\mathbf{h}_i(x,E_i)$. It would thus not be possible to address multiple exact energies with a single quantum simulation of $\mathbf{h}_i(x,E_i)$, which is one key purpose of downfolding from a quantum information standpoint.


In terms of gaining an intuition about the origin of this no-downfolding property, one approach is to consider the case of two real parameters, $y$ and $z$, rather than a single complex argument:
\begin{equation}
    \mathbf{H}(y,z)\equiv \mathbf{A}+ y \mathbf{B} + z \mathbf{C}
\end{equation}
The resulting eigenvalues, $E_i(y,z)$, are real for Hermitian $\mathbf{A},\mathbf{B},\mathbf{C}$. As is well-known, in the absence of symmetry a Hamiltonian of this form admits degeneracies  $E_i(y^*,z^*)=E_j(y^*,z^*)$ at isolated points $(y^*,z^*)$ termed conical intersections \cite{Yarkony96,WorthCederbaum04,LevineMartinez07} or diabolical points \cite{BerryWilkinson84,Murakami07}. In this case, one can move continuously from one eigenvalue to another without resorting to complex parameters by passing through an intersection. Suppose every pair of adjacent eigenvalues $[i,i+1]$ had a point of intersection $(y^*_i,z^*_i)$. If this were the case, one could simply connect state $i$ to state $k$  ($k>i$) at $(y_0,z_0)$ by tracing a path, $y(t),z(t)$  that begins at $(y_0,z_0)$ through $(y^*_i,z^*_i)$,$(y^*_{i+1},z^*_{i+1})$,...
$(y^*_{k-1},z^*_{k-1})$ before finally returning to $(y_0,z_0)$. In this way, one would show that $E_i(y(s),z(s))$ and $E_k(y(s),z(s))$ are a single-valued, smooth function of $s$ but a multi-valued function of $y,z$ - which would again prove the impossibility of exact quantum downfolding. While this construction is physically appealing and certainly seems plausible, it is not clear what conditions on $\mathbf{A},\mathbf{B}$ and $\mathbf{C}$ are required to guarantee intersections between all adjacent eigenvalues. The construction in terms of a single complex variable makes this heuristic argument rigorous. 

As an interesting aside: it is widely recognized that, in the absence of symmetry, a single real parameter $x$ generally does not introduce degeneracies (the "non-crossing" rule\cite{Teller37}). The present results show that, in the absence of symmetry, the introduction of a single \emph{complex} parameter is sufficient to create degeneracies that link \emph{all} pairs of eigenvalues: the shuffling involved in monodromy is generated entirely by loops that pass around various intersection points (branch points) in the complex plane.

One can also obtain a physical picture of the present approach by appealing to various semiclassical approximations \cite{MillerGeorge72,Pechukas77} in which it has been noted that making time complex can allow classical trajectories to effectively transition between two different quantum states. In this letter, we use an external parameter, $x$, instead of time, $t$,  to drive transitions between states. Additionally, the evolution here is adiabatic instead of following an equation of motion. However the semiclassical picture is similar in spirit and provided substantial inspiration for the present approach.


The practical implications of these results are significant. For downfolding approaches on classical computers, it is now clear that one is typically attempting to approximate an $N$-valued function with $M$-valued approximants. This approach is only well defined when anchored to a particular point $x_0$ and cannot be globally extended to a solution for all $x$. Attempts to do so will result in divergence, the (dis)appearance of fictitious states and/or discontinuity of the solution. Indeed, such results are commonly encountered in practice for downfolding methods and are generically referred to as the intruder state problem \cite{EvangleistiMalrieu87}. It has often been assumed that intruder state behavior was a flaw of approximate downfolding methods, and that downfolding could be made intruder-state free by the judicious regularization, broadening and/or energy level shifts \cite{Vincent73,ChoeHirao01,WitekHirao02,Nikolić_2004,ChangWitek12,BattagliaLindh22,MoninoLoos22}. The present results demonstrate that intruder states cannot be globally removed from downfolded Hamiltonians and must instead be accepted as a feature of the approach.

For quantum simulation, the impacts are equally significant. Given that it is not possible to globally downfold quantum Hamiltonians, it seems much harder to avoid the situation in which large systems simply require large numbers of qubits for quantum simulation. On the one hand, this means that practical impacts of quantum computing may be somewhat further into the future than previously anticipated: one cannot rely on downfolding to squeeze larger systems onto smaller quantum computers. At the same time, quantum advantage is now easier to realize: with sufficiently large numbers of qubits, one can be more confident there is not be a global classical downfolding that would allow the same system to be simulated with existing classical resources.

This work also has a significant impact on the application of approximate downfolding. 
For a given $\mathbf{H}(x)$ and $M < N$ there will be better and worse approximations. But there some non-zero smallest possible error, $\epsilon$, one can obtain via downfolding - and this optimal accuracy is fixed. 
$\epsilon$ may be small enough for some situations. 
But in other situations one instead requires a tool where the accuracy can be controlled  - either because the precision requirements of the physical system under study are not known \emph{a priori} or because one is interested in proof-of-concept algorithm verification. It is clear from the present work that downfolding will have limited utility in these scenarios.

In the near future, it would be interesting to explore how the present approach impacts multi-fragment embedding \cite{KniziaChan12,Bootstrap16} in which a large Hamiltonian is decomposed not into a single smaller Hamiltonian, but multiple fragment Hamiltonians. Similarly, given that energy-dependent downfolding is possible, it will be important to further explore how Green's function techniques could be implemented on quantum hardware\cite{DMFT,LanZgid15}, noting that the resulting effective Hamiltonian is often not Hermitian. Finally, there is a need to examine whether similar limits apply to tensor network states \cite{PengWu20,DMRG,Orus14} where downfolding is performed in the bond dimension space, rather than in Hilbert space.

\begin{acknowledgments}
\emph{Acknowledgements}--This work was supported by the U.S. Department of Energy, Oﬃce of Science, National Quantum Information Science Research Centers, Co-Design Center for Quantum Advantage under Grant No. DE-SC0012704. The author would like to thank Alan Edelman, Ike Chuang and Aram Harrow for helpful discussion.
\end{acknowledgments}


\bibliography{NoDownfolding}

\end{document}